\documentclass[a4paper,twocolumn,11pt]{quantumarticle}
\pdfoutput=1
\usepackage[utf8]{inputenc}
\usepackage[english]{babel}
\usepackage[T1]{fontenc}
\usepackage{amsmath}
\usepackage{hyperref}

\usepackage{tikz}
\usepackage{lipsum}

\usepackage{enumerate}
\usepackage{amsmath}
\usepackage{booktabs}
\usepackage{array}
\usepackage{amsthm}
\usepackage{float}
\usepackage{dcolumn}
\usepackage{longtable}
\usepackage{graphicx}
\usepackage{epsfig}
\usepackage{epstopdf}
\usepackage{amssymb}
\usepackage{bm}
\usepackage{mathrsfs}
\usepackage{color}
\usepackage{cite}

\newtheorem{theorem}{Theorem}
\theoremstyle{remark}
\newtheorem{definition}{Definition}
\newtheorem{lemma}{Lemma}

\newtheorem{example}{Example}

\begin{document}

\title{Strongest nonlocal sets with minimum cardinality in multipartite systems}

\author{Hong-Run Li}
\affiliation{School of Mathematical Sciences, Hebei Normal University, Shijiazhuang 050024, China}

\author{Hui-Juan Zuo}
\email{huijuanzuo@163.com}
\affiliation{School of Mathematical Sciences, Hebei Normal University, Shijiazhuang 050024, China}

\author{Fei Shi}
\affiliation{QICI Quantum Information and Computation Initiative, Department of Computer Science, The University of Hong Kong, Pokfulam Road 999077, Hong Kong, China}

\author{Shao-Ming Fei}
\affiliation{School of Mathematical Sciences, Capital Normal University, Beijing 100048, China}

\maketitle

\begin{abstract}
  Quantum nonlocality based on state discrimination describes the global property of the set of orthogonal states and has a wide range of applications in quantum cryptographic protocols. Strongest nonlocality is the strongest form of quantum nonlocality recently presented in multipartite quantum systems: a set of orthogonal multipartite quantum states is strongest nonlocal if the only orthogonality-preserving local measurements on the subsystems in every bipartition are trivial. In this work, we found a construction of strongest nonlocal sets in $\mathbb{C}^{d_{1}}\otimes \mathbb{C}^{d_{2}}\otimes \mathbb{C}^{d_{3}}$ $(2\leq d_{1}\leq d_{2}\leq d_{3})$ of size $d_2d_3+1$ without stopper states. Then we obtain the strongest nonlocal sets in four-partite systems with $d^3+1$ orthogonal states in $\mathbb{C}^d\otimes \mathbb{C}^{d}\otimes \mathbb{C}^{d}\otimes \mathbb{C}^{d}$ $(d\geq2)$ and $d_{2}d_{3}d_{4}+1$ orthogonal states in $\mathbb{C}^{d_{1}}\otimes \mathbb{C}^{d_{2}}\otimes \mathbb{C}^{d_{3}}\otimes \mathbb{C}^{d_{4}}$ $(2\leq d_{1}\leq d_{2}\leq d_{3}\leq d_{4})$. Surprisingly, the number of the elements in all above constructions perfectly reaches the recent conjectured lower bound and reduces the size of the strongest nonlocal set in $\mathbb{C}^{d}\otimes \mathbb{C}^{d}\otimes \mathbb{C}^{d}\otimes \mathbb{C}^{d}$ of [\href{https://doi.org/10.1103/PhysRevA.108.062407}{Phys. Rev. A \textbf{108}, 062407 (2023)}] by $d-2$. In particular, the general optimal construction of the strongest nonlocal set in four-partite system is completely solved for the first time, which further highlights the theory of quantum nonlocality from the perspective of state discrimination.
\end{abstract}

\section{Introduction}\label{sec1}
Quantum nonlocality is one of the most fundamental feature in quantum mechanics. The usual Bell nonlocality is detected by the violation of Bell inequalities by quantum entangled states. In 1999, Bennett {\it et al}. \cite{Bennett1999} first presented an example of locally indistinguishable orthogonal product basis in $\mathbb{C}^{3}\otimes \mathbb{C}^{3}$, exhibiting the phenomenon of ``quantum nonlocality without entanglement". A set of orthogonal quantum states is locally indistinguishable if it is not possible to distinguish the states under local operations and classical communications (LOCC). Such nonlocality based on state discrimination is different from Bell nonlocality. Since then, the nonlocality based on state discrimination has been extensively studied \cite{Walgate2000,Ghosh2001,Walgate2002,Horodecki2003,Fan2004,Niset2006,Zhang2014,Li2015,Yu2015,Wang2015,Xu2016,WangYL2017,Zhang2017,Li2018,Halder2018,Jiang2020,Zuo2021,Xu2021,Zhen2022}, with wide applications in quantum cryptographic protocols such as quantum data hiding \cite{Terhal2001, Divincenzo2002, Eggeling2002} and quantum secret sharing \cite{Hillery1999, Rahaman2015, Wang2017}.

In 2019, Halder {\it et al}. \cite{Halder2019} introduced the concept of local irreducibility and strong nonlocality. A set of orthogonal quantum states is said to be locally irreducible, if it is not possible to eliminate one or more states from the set by orthogonality-preserving local measurements (OPLMs). A locally irreducible set must be a locally indistinguishable set, but the converse is not true except when it contains only three orthogonal pure states. A set of orthogonal quantum states is said to be strongly nonlocal if it is locally irreducible in every bipartition. They also gave the examples of strongly nonlocal orthogonal product sets in $\mathbb{C}^{3}\otimes \mathbb{C}^{3}\otimes \mathbb{C}^{3}$ and $\mathbb{C}^{4}\otimes \mathbb{C}^{4}\otimes \mathbb{C}^{4}$. Then, the results on strongly nonlocal sets have been obtained successively \cite{Zhang2019, Shi2020, Yuan2020, Wang2021, Shi2021, He2022, Shi2022, Hu2024, Shi2023, Li2023, Shi6619}. Walgate and Hardy mentioned that in any locally distinguishable protocol, one of the parties must go first and whoever goes first must be able to perform some nontrivial OPLMs. Accordingly, trivial OPLMs are sufficient to prove the nonlocality, which has been frequently utilized as a technique for detecting strong nonlocality of a set.

The notion of local stability was proposed in 2023, which is a stronger form than local irreducibility. An orthogonal set of pure states in multipartite quantum systems is said to be locally stable if the only possible OPLMs on the subsystems are trivial. A set of orthogonal states is strongest nonlocal if the only OPLMs on the subsystems in every bipartition are trivial \cite{Shi6619}. A conjecture on the bounds of locally stable sets and strongest nonlocal sets was proposed, i.e., in multipartite systems $\otimes _{i=1}^N \mathbb{C}^{d_i}$, the size of locally stable set is $\max_i\{d_i+1\}$ and the size of strongest nonlocal set is $\max_i\{ \hat{d_i}+1\}$, where $\hat{d_i}=(\Pi_{j=1}^{N}d_j)/d_i$.

Cao {\it et al}. provided the construction of locally stable sets that reach the lower bound of the cardinality \cite{Cao2023}. So far, the strongest nonlocal sets with the minimum cardinality is obtained only in tripartite systems \cite{Zhen2024}, but the direct construction cannot be generalized to $N$-partite systems $(N\geq4)$. It is generally deemed that a construction reaching the lower bound must contain a stopper state. In fact, that is not the case. In this paper, we find the strongest nonlocal sets with the minimum cardinality in $\mathbb{C}^{d_1}\otimes \mathbb{C}^{d_2}\otimes \mathbb{C}^{d_3} $ without stopper states. In particular, we obtain the strongest nonlocal sets reaching the lower bound in $\mathbb{C}^{d_1}\otimes \mathbb{C}^{d_2}\otimes \mathbb{C}^{d_3}\otimes \mathbb{C}^{d_4}$ for the first time.

The paper is organized as follows. In Sec. \ref{sec2}, we introduce the definitions and notations used in this paper. In Sec. \ref{sec3}, the strongest nonlocal sets without stopper states are presented in $\mathbb{C}^{d_1}\otimes \mathbb{C}^{d_2}\otimes \mathbb{C}^{d_3}$. In Sec. \ref{sec4}, we discuss the construction of the strongest nonlocal sets in detail in $\mathbb{C}^{d}\otimes \mathbb{C}^{d}\otimes \mathbb{C}^{d}\otimes \mathbb{C}^{d}$ and $\mathbb{C}^{d_1}\otimes \mathbb{C}^{d_2}\otimes \mathbb{C}^{d_3}\otimes \mathbb{C}^{d_4}$. Finally, we discuss and summarize in Sec. \ref{sec5}.

\section{Preliminaries}\label{sec2}

Throughout this paper, we assume that $\omega_n = e^{\frac{2\pi \sqrt{-1}}{n}}$ and $\{|0\rangle, |1\rangle, \cdots, |d-1\rangle \}$ is a computational basis of $\mathbb{C}^{d}$. A positive operator-valued measure (POVM) on Hilbert space $\mathcal{H}$ is a set of positive semidefinite operators $\{ E_m =M_m^{\dagger}M_m \}$ such that $\sum_m E_m=I_{\mathcal{H}}$, where $I_{\mathcal{H}}$ is the indentity operator on $\mathcal{H}$. A measurement is \textit{trivial} if each POVM element $E_m$ is proportional to the identity operator, i.e., $E_m\propto I_{\mathcal{H}}$. Otherwise, the measurement is called nontrivial. A measurement is \textit{orthogonal-preserving local measurement} (OPLM) if the post-measurement states keep the mutually orthogonality.

\begin{definition}(Strongest Nonlocality)
 A set of orthogonal multipartite quantum states is \textit{strongest nonlocal} if the only OPLMs on the subsystems in every bipartition are trivial.
\end{definition}
If the only OPLMs on party $B_i=\{A_1A_2 \cdots A_N \} \setminus {A_i}$ are trivial for any $1\le i \le N$, then the set in $\mathcal{H}_{A_1}\otimes \mathcal{H}_{A_2}\otimes\cdots \otimes\mathcal{H}_{A_N}$ is of the strongest nonlocality \cite{Shi6619}. In other words, we only need to show that the only OPLMs performed by any $N-1$ parties are trivial. In this paper, we mainly discuss the strongest nonlocal sets in $N$-parties $(N=3,4)$.  Note that the sets we construct have perfectly symmetric structure, so it is sufficient to prove the strongest nonlocality that the OPLMs on subsystem $BC$ and $BCD$ are trivial. We first recall two lemmas presented in Ref. \cite{Shi6619}.

\begin{lemma}(Block Zeros Lemma)\label{lem:zero}
Let $E=(a_{i,j})_{i,j\in \mathbb{Z}_{d}}$ be the $d \times d$ matrix representation of an operator $E$ under the basis $\mathcal{B}=\{ |0\rangle, |1\rangle, \cdots, |d-1\rangle\}$. Given two nonempty disjoint subsets $\mathcal{S}$ and $\mathcal{T}$ of $\mathcal{B}$, let $\{|\psi_i\rangle \}_{i\in \mathbb{Z}_s}$ and $\{|\phi_j\rangle \}_{j\in \mathbb{Z}_t}$ be two orthogonal sets spanned by $\mathcal{S}$ and $\mathcal{T}$, respectively, where $s=|\mathcal{S}|$ and $t=|\mathcal{T}|$. If $\langle \psi_i| E |\phi_j\rangle=0$ for any $i \in \mathbb{Z}_{s}$, $j \in \mathbb{Z}_{t}$, then $\langle x| E |y \rangle= \langle y| E |x \rangle=0$ for $|x\rangle \in \mathcal{S}$ and $|y\rangle \in \mathcal{T}$.
\end{lemma}

\begin{lemma}(Block Trivial Lemma)\label{lem:trivial}
Let $E=(a_{i,j})_{i,j\in \mathbb{Z}_{d}}$ be the $d \times d$ matrix representation of an operator $E$ under the basis $\mathcal{B}=\{ |0\rangle, |1\rangle, \cdots, |d-1\rangle\}$. Given a nonempty subset $\mathcal{S}$ of $\mathcal{B}$, let $\{|\psi_i\rangle \}_{i\in \mathbb{Z}_s}$ be an orthogonal set spanned by $\mathcal{S}$. Assume that $\langle \psi_i| E |\psi_j \rangle=0$ for any $i \ne j \in \mathbb{Z}_s$. If there exists a state $|x\rangle \in \mathcal{S}$ such that $\langle x| E |y\rangle=0$ for all $|y\rangle \in \mathcal{S} \setminus \{|x\rangle \} $ and $\langle x|\psi_j\rangle \ne 0$ for any $j\in \mathbb{Z}_s$, then $\langle y| E |z\rangle=0$ and $\langle y| E |y \rangle=\langle z| E |z\rangle$ for all $|y\rangle, |z\rangle \in \mathcal{S}$ with $|y\rangle \ne |z\rangle$.
\end{lemma}

\section{Strongest Nonlocal Sets in Tripartite Systems}\label{sec3}
In this section, we present the strongest nonlocal sets with the minimum cardinality in $\mathbb{C}^{d_1}\otimes \mathbb{C}^{d_2}\otimes \mathbb{C}^{d_3}$. First, we consider the construction in $\mathbb{C}^{d}\otimes \mathbb{C}^{d}\otimes \mathbb{C}^{d}$. Here each state is shared by A, B, and C. Define
    \begin{equation*}\small
      \begin{aligned}
      &\mathcal{B}_{1}=\{ |\alpha_1\rangle = |100\rangle, |\alpha_2\rangle = |001\rangle, |\alpha_3\rangle = |010\rangle\},\\
      &\mathcal{B}_{2}=\{ \frac{1}{\sqrt{2}} (|00i\rangle +|(i-1)~(i-1)~0\rangle)~|~ 2\le i \le d-1\},\\
      &\mathcal{B}_{3}=\{ \frac{1}{\sqrt{2} } (|0i0\rangle +|(i-1)~0~(i-1)\rangle)~|~ 2\le i \le d-1\},\\
      &\mathcal{B}_{4}=\{ \frac{1}{\sqrt{2} } (|0~(i-1)~(i-1)\rangle +|i00\rangle)~|~ 2\le i \le d-1\},\\
      &\mathcal{B}_{5}=\{ \frac{1}{\sqrt{3} } (|0ij\rangle +|ij0\rangle +|j0i\rangle)~|~1\le i,j \le d-1\\
      &~~~~~~~~~ \backslash(i,j)=(1,1), \cdots, (d-2,d-2) \}.
      \end{aligned}
    \end{equation*}
\begin{theorem}
In $\mathbb{C}^{d}\otimes \mathbb{C}^{d}\otimes \mathbb{C}^{d}$ $(d\geq 2)$, the set $\{|\psi_{i} \rangle\}_{i=0}^{d^2}$ below is a strongest nonlocal set of size $d^2+1$,
\begin{equation}\label{eq:ddd1}
  \begin{aligned}
  &|\psi_{0} \rangle= |\alpha_{0} \rangle, \\
  &|\psi_{i} \rangle= \sum_{j=1}^{d^2} \omega_{d^2}^{ij}|\alpha_{j} \rangle, ~~~ 1\le i \le d^2,\\
  \end{aligned}
\end{equation}
where
\begin{equation}\label{eq:ddd2}
  \begin{aligned}
   &|\alpha_{0} \rangle =|000\rangle,\\
   &\{ |\alpha_{j} \rangle\}_{j=1}^{d^2}=\mathcal{B}_{1}\cup \mathcal{B}_{2}\cup \mathcal{B}_{3}\cup \mathcal{B}_{4}\cup \mathcal{B}_{5}.
  \end{aligned}
\end{equation}
\end{theorem}

\begin{proof}
 Suppose that B and C perform a joint OPLM $\{E=M^{\dagger}M\}$. The post-measurement states should be mutually orthogonal, i.e., $\langle \psi_{i}| I\otimes E |\psi_{j} \rangle = 0$ for $0\le i\ne j \le d^2$.

\begin{figure*}[ht!]
\centerline{\includegraphics[width=\textwidth, keepaspectratio]{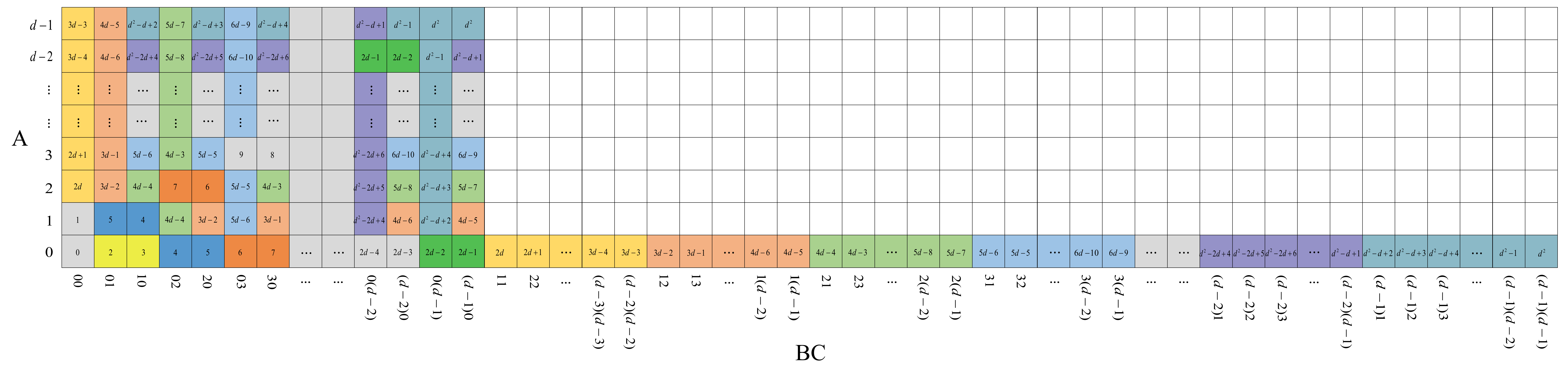}}
\caption{The corresponding $d\times d^2$ grid of $\{|\alpha_{i} \rangle\}_{i=0}^{d^2}$ given by Eq. (\ref{eq:ddd2}) in $A|BC$ bipartition. Every grid has an index $(i, mn)$, where $i$ is the row index in $A$ part and $mn$ is the column index in $BC$ part. For example, $|\alpha_4 \rangle$ corresponds to the cell set $\{(0,02), (1, 10)\}$. }
\label{fig1}
\end{figure*}
 Note that $\{|\psi_{0} \rangle\}$ is spanned by $\{|\alpha_{0}\rangle\}$ and $\{|\psi_{i} \rangle\}_{i=1}^{d^2}$ is spanned by $\{|\alpha_{j} \rangle\}_{j=1}^{d^2}$. For clarity, the structure of $\{|\alpha_{i} \rangle\}_{i=0}^{d^2}$ in $A|BC$ bipartition is shown in Fig. \ref{fig1}. As $\langle \psi_{0}| I\otimes E |\psi_{i} \rangle = 0$ for $1 \le i \le d^2$, applying Lemma \ref{lem:zero} to $\{|\psi_{0}\rangle\}$ and $\{|\psi_{i} \rangle\}_{i=1}^{d^2}$, we have $\langle \alpha_{0}| I\otimes E |\alpha_{j} \rangle =\langle \alpha_{j}| I\otimes E |\alpha_{0} \rangle = 0$ for $1\le j \le d^2$, which implies that
 \begin{equation*}
        \begin{aligned}
            \langle 00|E|mn\rangle =\langle mn|E|00\rangle = 0,
        \end{aligned}
    \end{equation*}
 where $(m,n)\in \mathbb{Z}_d \times \mathbb{Z}_d\backslash (0,0)$. Hence, we have
    \begin{equation*}
        \begin{aligned}
            &\langle \alpha_{1}| I\otimes E |\alpha_{j} \rangle = 0, ~~2\le j\le d^2,\\
            &\langle \alpha_1 |\psi_j \rangle \ne 0, ~~1\le j \le d^2.
        \end{aligned}
    \end{equation*}
 Applying Lemma \ref{lem:trivial} to $\{|\psi_{i} \rangle\}_{i=1}^{d^2}$, we obtain
    \begin{equation*}
     \begin{aligned}
            &\langle \alpha_{x}| I\otimes E |\alpha_{y} \rangle = 0, ~~1\le x\ne y\le d^2,\\
            &\langle \alpha_{x}| I\otimes E |\alpha_{x} \rangle = \langle \alpha_{y}| I\otimes E |\alpha_{y} \rangle, ~~1\le x\ne y \le d^2.
      \end{aligned}
    \end{equation*}

    First, we prove that the off-diagonal entries of the matrix $E$ are zeros, namely, $\langle ij |E|mn \rangle=0$, where $(m,n)\neq (i,j)$ in the proof below.\\
    \textbf{Step 1}
    From $\langle001| I\otimes E |\alpha_y\rangle =\langle010| I\otimes E |\alpha_y\rangle = 0$, we have
    \begin{equation*}
        \begin{aligned}
            \langle01| E |mn\rangle =\langle10| E |mn\rangle = 0.
        \end{aligned}
    \end{equation*}
    \textbf{Step 2}
    Since $\langle10| E |mn\rangle = 0$, $\frac{1}{\sqrt{2}}(\langle 002|+\langle 110|) I\otimes E |\alpha_{y} \rangle =\frac{1}{\sqrt{2}} (\langle 002| I\otimes E |\alpha_{y} \rangle+ \langle 110| I\otimes E |\alpha_{y} \rangle)=0$ implies that $\langle 002| I\otimes E |\alpha_{j} \rangle=0$, i.e.,
    \begin{equation*}
        \begin{aligned}
            \langle 02| E |mn \rangle=0.
        \end{aligned}
    \end{equation*}
    In the same way, $\frac{1}{\sqrt{2}}(\langle 020|+\langle 101|) I\otimes E |\alpha_{j} \rangle =0$ implies that
    \begin{equation*}
        \begin{aligned}
            \langle20| E |mn\rangle= 0.
        \end{aligned}
    \end{equation*}
    \textbf{Step \bm{$k~(3\le k \le d-1)$}}
    From the ($k-1$)-th step, we get $\langle 0~(k-1)| E |mn\rangle =\langle (k-1)~0| E |mn\rangle =0$. Thus $\frac{1}{\sqrt{2}}(\langle 00k|+\langle (k-1)~(k-1)~0|) E |\alpha_y\rangle =\frac{1}{\sqrt{2}}(\langle 0k0|+\langle (k-1)~0~(k-1)|) E |\alpha_y\rangle=0$ means that
    \begin{equation*}
        \begin{aligned}
            \langle 0k| E |mn\rangle =\langle k0| E |mn\rangle=0, ~3\le k \le d-1.
        \end{aligned}
    \end{equation*}
    \textbf{Step \bm{$d$}}
    According to the above results, for the states in $\mathcal{B}_{4}$ and $\mathcal{B}_{5}$ we have
    \begin{equation*}
        \begin{aligned}
           \langle ij| E |mn\rangle = 0, ~1\le i,j \le d-1.
        \end{aligned}
    \end{equation*}

    Next, we prove that all the diagonal entries of the matrix $E$ are equal.\\
    \textbf{Step 1}
    From $\langle100| I\otimes E |100\rangle =\langle001| I\otimes E |001\rangle =\langle010| I\otimes E |010\rangle $, we have
    \begin{equation*}
        \begin{aligned}
          \langle00| E |00\rangle =\langle01| E |01\rangle=\langle10| E |10\rangle.
        \end{aligned}
    \end{equation*}
    \textbf{Step 2}
    As $\langle00| E |00\rangle =\langle10| E |10\rangle$, $\langle100| I\otimes E |100\rangle =\frac{1}{2}(\langle 002|+\langle 110|) I\otimes E (|002\rangle+ |110\rangle) = \frac{1}{2}(\langle 002| I\otimes E |002\rangle + \langle 110| I\otimes E |110\rangle)$ implies that
    \begin{equation*}
        \begin{aligned}
          \langle00| E |00\rangle =\langle02| E |02\rangle.
        \end{aligned}
    \end{equation*}
    Similarly, $\langle100| I\otimes E |100\rangle =\frac{1}{2}(\langle 020|+\langle 101|) I\otimes E (|020\rangle+ |101\rangle)$ means that
    \begin{equation*}
        \begin{aligned}
          \langle00| E |00\rangle =\langle20| E |20\rangle.
        \end{aligned}
    \end{equation*}
    \textbf{Step \bm{$k~(3\le k \le d-1)$}}
    From the ($k-1$)-th step, we get $\langle00| E |00\rangle =\langle0~(k-1)| E |0~(k-1)\rangle=\langle (k-1)~0| E |(k-1)~0\rangle$. Thus $\langle100| I\otimes E |100\rangle =\frac{1}{2}(\langle 00k|+\langle (k-1)~(k-1)~0|) I\otimes E (| 00k\rangle+|(k-1)~(k-1)~0\rangle)=\frac{1}{2}(\langle 0k0|+\langle (k-1)~0~(k-1)|) I\otimes E ( |0k0\rangle+|(k-1)~0~(k-1)\rangle )$ implies that
     \begin{equation*}
        \begin{aligned}
           \langle00| E |00\rangle =\langle 0k| E |0k\rangle=\langle k0| E |k0\rangle, ~3\le k \le d-1.
        \end{aligned}
    \end{equation*}
    \textbf{Step \bm{$d$}}
    Based on the above results, considering the states in $\mathcal{B}_{4,5}$ and $|100\rangle$, we have
    \begin{equation*}
        \begin{aligned}
           \langle00| E |00\rangle =\langle ij| E |ij\rangle, ~1\le i,j \le d-1.
        \end{aligned}
    \end{equation*}

    Thus $E\propto I$. This completes the proof.
\end{proof}

Now we generalize the construction to general tripartite systems $\mathbb{C}^{d_{1}}\otimes \mathbb{C}^{d_{2}}\otimes \mathbb{C}^{d_{3}}$ $(2\leq d_{1}\leq d_{2}\leq d_{3})$. Define
\begin{equation*}\small
  \begin{aligned}
    &\mathcal{B}_{1}=\{ |\alpha_1\rangle = |100\rangle, |\alpha_2\rangle = |001\rangle, |\alpha_3\rangle = |010\rangle\},\\
    &\mathcal{B}_{2}=\{ \frac{1}{\sqrt{2} } (|00i\rangle +|(i-1)~(i-1)~0\rangle)~|~2\le i \le d_1 \},\\
    &\mathcal{B}_{3}=\{ \frac{1}{\sqrt{2} } (|0i0\rangle +|(i-1)~0~(i-1)\rangle)~|~2\le i \le d_1 \},\\
    &\mathcal{B}_{4}=\{ |00i\rangle ~|~d_1+1\le i \le d_3-1 \},\\
    &\mathcal{B}_{5}=\{ |0i0\rangle ~|~d_1+1\le i \le d_2-1 \},\\
    &\mathcal{B}_{6}=\{ \frac{1}{\sqrt{2} } (|0~(i-1)~(i-1)\rangle +|i00\rangle)~|~2\le i \le d_1-1 \},\\
    &\mathcal{B}_{7}= \{ |0~(i-1)~(i-1)\rangle ~|~d_1\le i \le d_2 \},\\
    &\mathcal{B}_{8}= \{ \frac{1}{\sqrt{3} } (|0ij\rangle +|ij0\rangle +|j0i\rangle)~|~1\le i,j \le d_1-1\\
    &~~~~~~~~ \backslash(i,j) = (1,1), \cdots,(d_1-1,d_1-1) \},\\
    &\mathcal{B}_{9}=\{ \frac{1}{\sqrt{2} } (|0ij\rangle +|ij0\rangle)~|~1\le i \le d_1-1,\\
    &~~~~~~~~~d_1 \le j \le d_2-1 \},\\
    &\mathcal{B}_{10}=\{ \frac{1}{\sqrt{2} } (|0ji\rangle +|i0j\rangle)~|~1\le i \le d_1-1,\\
    &~~~~~~~~~d_1 \le j \le d_2-1 \},\\
    &\mathcal{B}_{11}=\{ \frac{1}{\sqrt{2} } (|0ij\rangle +|i0j\rangle)~|~1\le i \le d_1-1, \\
    &~~~~~~~~~d_2 \le j \le d_3-1 \},\\
    &\mathcal{B}_{12}=\{ |0ij\rangle ~|~d_1\le i \le d_2-1, ~d_2 \le j \le d_3-1\},\\
    &\mathcal{B}_{13}=\{ |0ij\rangle ~|~d_1\le i,j \le d_2-1 \backslash(i,j) =(d_1,d_1), \cdots,\\
    &~~~~~~~~~(d_2-1,d_2-1)\}.\\
    \end{aligned}
\end{equation*}

\begin{theorem}\label{th:d1d2d3}
    In $\mathbb{C}^{d_{1}}\otimes \mathbb{C}^{d_{2}}\otimes \mathbb{C}^{d_{3}}$ $(2 \le d_1 \le d_2 \le d_3)$, the set $\{|\psi_{i} \rangle\}_{i=0}^{d_2d_3}$ below is a strongest nonlocal set of size $d_2d_3+1$,
    \begin{equation}\label{eq:d1d2d3}
      \begin{aligned}
         |\psi_{0} \rangle &= |\alpha_{0} \rangle, \\
        |\psi_{i} \rangle &= \sum_{j=1}^{d_2d_3} \omega_{d_2d_3}^{ij}|\alpha_{j} \rangle, ~~~ 1\le i \le d_2d_3,\\
      \end{aligned}
    \end{equation}
where
\begin{equation}
  \begin{aligned}
   &|\alpha_{0} \rangle=|000\rangle,\\
   &\{ |\alpha_{j} \rangle\}_{j=1}^{d_2d_3}=\mathcal{B}_{1}\cup \mathcal{B}_{2}\cup \mathcal{B}_{3}\cup \cdots \cup \mathcal{B}_{13}.
   \end{aligned}
\end{equation}
\end{theorem}

The proof of Theorem \ref{th:d1d2d3} is given in Appendix \ref{append:A}. Compared with the  strongest nonlocal set in $\mathbb{C}^{d_1}\otimes \mathbb{C}^{d_2}\otimes \mathbb{C}^{d_3}$ in Ref.~\cite{Zhen2024}, our strongest nonlocal set does not contain the stopper state $|S \rangle=(\sum_{i=0}^{d_1-1}|i \rangle)\otimes(\sum_{j=0}^{d_2-1}|j \rangle)\otimes(\sum_{k=0}^{d_3-1}|k \rangle)$.

\section{Strongest Nonlocal Sets in Four-partite Systems}\label{sec4}
In this section, we first present an orthogonal set of size 28 with the strongest nonlocality in $\mathbb{C}^{3}\otimes \mathbb{C}^{3}\otimes \mathbb{C}^{3}\otimes \mathbb{C}^{3}$. Then we give the strongest nonlocal sets with the minimum cardinality in $\mathbb{C}^{d_{1}}\otimes \mathbb{C}^{d_{2}}\otimes \mathbb{C}^{d_{3}}\otimes \mathbb{C}^{d_{4}}$. Here each state is shared by A, B, C, and D.

\begin{figure*}[ht!]
\centerline{\includegraphics[width=\textwidth, keepaspectratio]{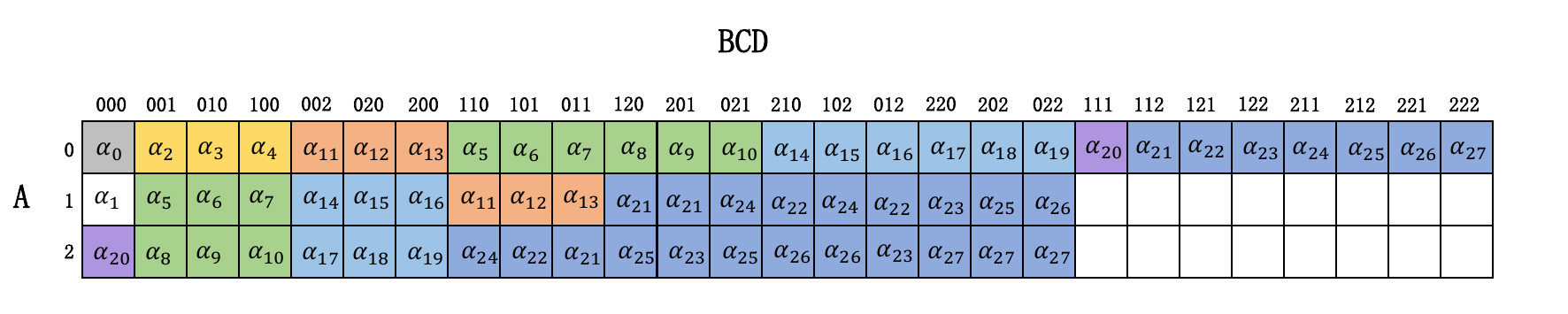}}
\caption{The corresponding $3\times 27$ grid of $\{|\alpha_{i} \rangle\}_{i=0}^{27}$ given by Eq. (\ref{eq:33332}) in $A|BCD$ bipartition. Every grid has an index $(i, mnl)$, where $i$ is the row index in part $A$ and $mnl$ is the column index in joint part $BCD$. For example, $|\alpha_{11} \rangle$ corresponds to the cell set $\{(0,002), (1, 110)\}$.}
\label{fig2}
\end{figure*}

\begin{example}
    In $\mathbb{C}^{3}\otimes \mathbb{C}^{3}\otimes \mathbb{C}^{3}\otimes \mathbb{C}^{3}$, the set $\{|\psi_{i} \rangle\}_{i=0}^{27}$ below is a strongest nonlocal set of size 28,
    \begin{equation}\label{eq:33331}
      \begin{aligned}
       |\psi_{0} \rangle &= |\alpha_{0} \rangle, \\
       |\psi_{i} \rangle &= \sum_{j=1}^{27} \omega_{27}^{ij}~|\alpha_{j} \rangle, ~~~ 1\le i \le 27,
      \end{aligned}
    \end{equation}
where
\begin{equation}\label{eq:33332}
\begin{aligned}
|\alpha_0\rangle &= |0000\rangle, \\
|\alpha_1\rangle &= |1000\rangle, \\
|\alpha_2\rangle &= |0001\rangle, \\
|\alpha_3\rangle &= |0010\rangle, \\
|\alpha_4\rangle &= |0100\rangle, \\
|\alpha_5\rangle &= \frac{1}{\sqrt{2} } (|0110\rangle +|1001\rangle), \\
|\alpha_6\rangle &= \frac{1}{\sqrt{2} } (|0101\rangle +|1010\rangle), \\
|\alpha_7\rangle &= \frac{1}{\sqrt{2} } (|0011\rangle +|1100\rangle), \\
|\alpha_8\rangle &= \frac{1}{\sqrt{2} } (|0120\rangle +|2001\rangle), \\
|\alpha_9\rangle &= \frac{1}{\sqrt{2} } (|0201\rangle +|2010\rangle), \\
|\alpha_{10}\rangle &= \frac{1}{\sqrt{2} } (|0021\rangle +|2100\rangle), \\
|\alpha_{11}\rangle &= \frac{1}{\sqrt{2} } (|0002\rangle +|1110\rangle), \\
|\alpha_{12}\rangle &= \frac{1}{\sqrt{2} } (|0020\rangle +|1101\rangle), \\
|\alpha_{13}\rangle &= \frac{1}{\sqrt{2} } (|0200\rangle +|1011\rangle), \\
|\alpha_{14}\rangle &= \frac{1}{\sqrt{2} } (|0210\rangle +|1002\rangle), \\
|\alpha_{15}\rangle &= \frac{1}{\sqrt{2} } (|0102\rangle +|1020\rangle), \\
|\alpha_{16}\rangle &= \frac{1}{\sqrt{2} } (|0012\rangle +|1200\rangle), \\
|\alpha_{17}\rangle &= \frac{1}{\sqrt{2} } (|0220\rangle +|2002\rangle), \\
|\alpha_{18}\rangle &= \frac{1}{\sqrt{2} } (|0202\rangle +|2020\rangle), \\
|\alpha_{19}\rangle &= \frac{1}{\sqrt{2} } (|0022\rangle +|2200\rangle), \\
|\alpha_{20}\rangle &= \frac{1}{\sqrt{2} } (|0111\rangle +|2000\rangle), \\
\end{aligned}
\end{equation}
\begin{equation*}
\begin{aligned}
|\alpha_{21}\rangle &= \frac{1}{2} (|0112\rangle +|1120\rangle +|1201\rangle +|2011\rangle), \\
|\alpha_{22}\rangle &= \frac{1}{2} (|0121\rangle +|1210\rangle +|2101\rangle +|1012\rangle), \\
|\alpha_{23}\rangle &= \frac{1}{2} (|0122\rangle +|1220\rangle +|2201\rangle +|2012\rangle), \\
|\alpha_{24}\rangle &= \frac{1}{2} (|0211\rangle +|2110\rangle +|1102\rangle +|1021\rangle), \\
\end{aligned}
\end{equation*}
\begin{equation*}
\begin{aligned}
|\alpha_{25}\rangle &= \frac{1}{2} (|0212\rangle +|2120\rangle +|1202\rangle +|2021\rangle), \\
|\alpha_{26}\rangle &= \frac{1}{2} (|0221\rangle +|2210\rangle +|2102\rangle +|1022\rangle), \\
|\alpha_{27}\rangle &= \frac{1}{2} (|0222\rangle +|2220\rangle +|2202\rangle +|2022\rangle).
\end{aligned}
\end{equation*}
\end{example}

\begin{proof}
     Suppose that B, C, and D perform a joint OPLM $\{E=M^{\dagger}M\}$. The post-measurement states should be mutually orthogonal, i.e., $\langle \psi_{i}| I\otimes E |\psi_{j} \rangle = 0$ for $0\le i\ne j \le 27$.

    Note that $\{|\psi_{0}\rangle\}$ is spanned by $\{|\alpha_{0}\rangle\}$ and $\{|\psi_{i} \rangle\}_{i=1}^{27}$ is spanned by $\{|\alpha_{j} \rangle\}_{j=1}^{27}$. The structure of $\{|\alpha_{i} \rangle\}_{i=0}^{27}$ in $A|BCD$ bipartition is shown in Fig. \ref{fig2}. As $\langle \psi_{0}| I\otimes E |\psi_{i} \rangle = 0$ for $1 \le i \le 27$, applying Lemma \ref{lem:zero} to $\{|\psi_{0}\rangle\}$ and $\{|\psi_{i} \rangle\}_{i=1}^{27}$, we obtain $\langle \alpha_{0}| I\otimes E |\alpha_{j} \rangle = 0$ for $1\le j \le 27$, which implies that
     \begin{equation*}\small
        \begin{aligned}
          \langle 000|E|mnl\rangle =\langle mnl|E|000 \rangle = 0,
        \end{aligned}
    \end{equation*}
    where $(m,n,l)\in \mathbb{Z}_3 \times \mathbb{Z}_3 \times \mathbb{Z}_3\backslash (0,0,0)$. Therefore,
    \begin{equation*}
        \begin{aligned}
            &\langle \alpha_{1}| I\otimes E |\alpha_{j} \rangle= 0, ~~2\le j\le 27,\\
            &\langle \alpha_1 |\psi_j \rangle\ne 0, ~~1\le j \le 27.
        \end{aligned}
    \end{equation*}
    Applying Lemma \ref{lem:trivial} to $\{|\psi_{i} \rangle\}_{i=1}^{27}$, we have
    \begin{equation*}
        \begin{aligned}
            &\langle \alpha_{x}| I\otimes E |\alpha_{y} \rangle= 0, ~~1\le x\ne y\le 27,\\
            &\langle \alpha_{x}| I\otimes E |\alpha_{x} \rangle= \langle \alpha_{y}| I\otimes E |\alpha_{y} \rangle, ~~1\le x\ne y \le 27.
        \end{aligned}
    \end{equation*}

    First, we prove that the off-diagonal entries of the matrix $E$ are zeros, i.e., $\langle ijk |E|mnl \rangle=0$, where $(m, n, l) \ne (i, j, k)$ in the next proof.\\
    \textbf{Step 1}
    From $\langle\alpha_{x}| I\otimes E |\alpha_y\rangle = 0$ for $2 \le x \le 4$, we have
    \begin{equation*}\small
        \begin{aligned}
          \langle001| E |mnl\rangle =\langle010| E |mnl\rangle =\langle100| E |mnl\rangle = 0.
        \end{aligned}
    \end{equation*}
    Then $\langle \alpha_{x}| I\otimes E |\alpha_{y} \rangle = 0$ for $5 \le x \le 10$  implies that
    \begin{equation*}\small
        \begin{aligned}
          &\langle110| E |mnl\rangle =\langle101| E |mnl\rangle =\langle011| E |mnl\rangle\\
          = &\langle120| E |mnl\rangle =\langle201| E |mnl\rangle =\langle021| E |mnl\rangle = 0.
        \end{aligned}
    \end{equation*}

    \textbf{Step 2}
    Considering $|\alpha_{x}\rangle$ for $11\le x\le 13$, we have
    \begin{equation*}\small
        \begin{aligned}
          \langle002| E |mnl\rangle =\langle020| E |mnl\rangle =\langle200| E |mnl\rangle = 0.
        \end{aligned}
    \end{equation*}
    Then $\langle \alpha_{x}| I\otimes E |\alpha_{y} \rangle = 0$ for $14 \le x \le 19$ means that
    \begin{equation*}\small
        \begin{aligned}
          &\langle210| E |mnl\rangle =\langle102| E |mnl\rangle =\langle012| E |mnl\rangle\\
          =&\langle220| E |mnl\rangle =\langle202| E |mnl\rangle =\langle022| E |mnl\rangle = 0.
        \end{aligned}
    \end{equation*}
    \textbf{Step 3}
    Based on the above results, for $|\alpha_{x}\rangle$, $20 \le x \le 27$, we have
    \begin{equation*}\small
        \begin{aligned}
          &\langle111| E |mnl\rangle =\langle112| E |mnl\rangle =\langle121| E |mnl\rangle\\
          = &\langle122| E |mnl\rangle=\langle211| E |mnl\rangle =\langle212| E |mnl\rangle\\
          =&\langle221| E |mnl\rangle =\langle222| E |mnl\rangle = 0.
        \end{aligned}
    \end{equation*}

    Next, we prove that all the diagonal entries of the matrix $E$ are equal.\\
    \textbf{Step 1}
    From $\langle 1000| I\otimes E |1000 \rangle = \langle \alpha_{y}| I\otimes E |\alpha_{y} \rangle$ for $2\le y \le 4$, we have
    \begin{equation*}\small
        \begin{aligned}
           \langle000| E |000\rangle = \langle001| E |001\rangle = \langle010| E |010\rangle = \langle100| E |100\rangle.
        \end{aligned}
    \end{equation*}
    Then $\langle 1000| I\otimes E |1000 \rangle = \langle \alpha_{y}| I\otimes E |\alpha_{y} \rangle$ for $5\le y \le 10$ implies that
    \begin{equation*}\small
        \begin{aligned}
           \langle000| E |000\rangle &= \langle110| E |110\rangle = \langle101| E |101\rangle= \langle011| E |011\rangle\\
                                     & = \langle120| E |120\rangle = \langle201| E |201\rangle = \langle021| E |021\rangle.
        \end{aligned}
    \end{equation*}
    \textbf{Step 2}
    Considering $|\alpha_{y}\rangle$ for $11\le y \le 13$, we obtain
     \begin{equation*}\small
        \begin{aligned}
           \langle000| E |000\rangle = \langle002| E |002\rangle = \langle020| E |020\rangle = \langle200| E |200\rangle.
        \end{aligned}
    \end{equation*}
    Then $\langle 1000| I\otimes E |1000 \rangle = \langle \alpha_{y}| I\otimes E |\alpha_{y} \rangle$ for $14\le y \le 19$ means that
    \begin{equation*}\small
        \begin{aligned}
           \langle000| E |000\rangle &= \langle210| E |210\rangle = \langle102| E |102\rangle= \langle012| E |012\rangle\\
                                     &=\langle220| E |220\rangle = \langle202| E |202\rangle = \langle022| E |022\rangle.
        \end{aligned}
    \end{equation*}
    \textbf{Step 3}
    Based on the above results, for $|\alpha_{y}\rangle$, $20\le y \le 27$, we have
    \begin{equation*}\small
        \begin{aligned}
            &\langle000| E |000\rangle= \langle111| E |111\rangle = \langle112| E |112\rangle\\
            = &\langle121| E |121\rangle = \langle122| E |122\rangle= \langle211| E |211\rangle\\
            = &\langle212| E |212\rangle = \langle221| E |221\rangle = \langle222| E |222\rangle.
        \end{aligned}
    \end{equation*}
\begin{figure}[ht!]
	\centering
	\includegraphics[scale=0.1]{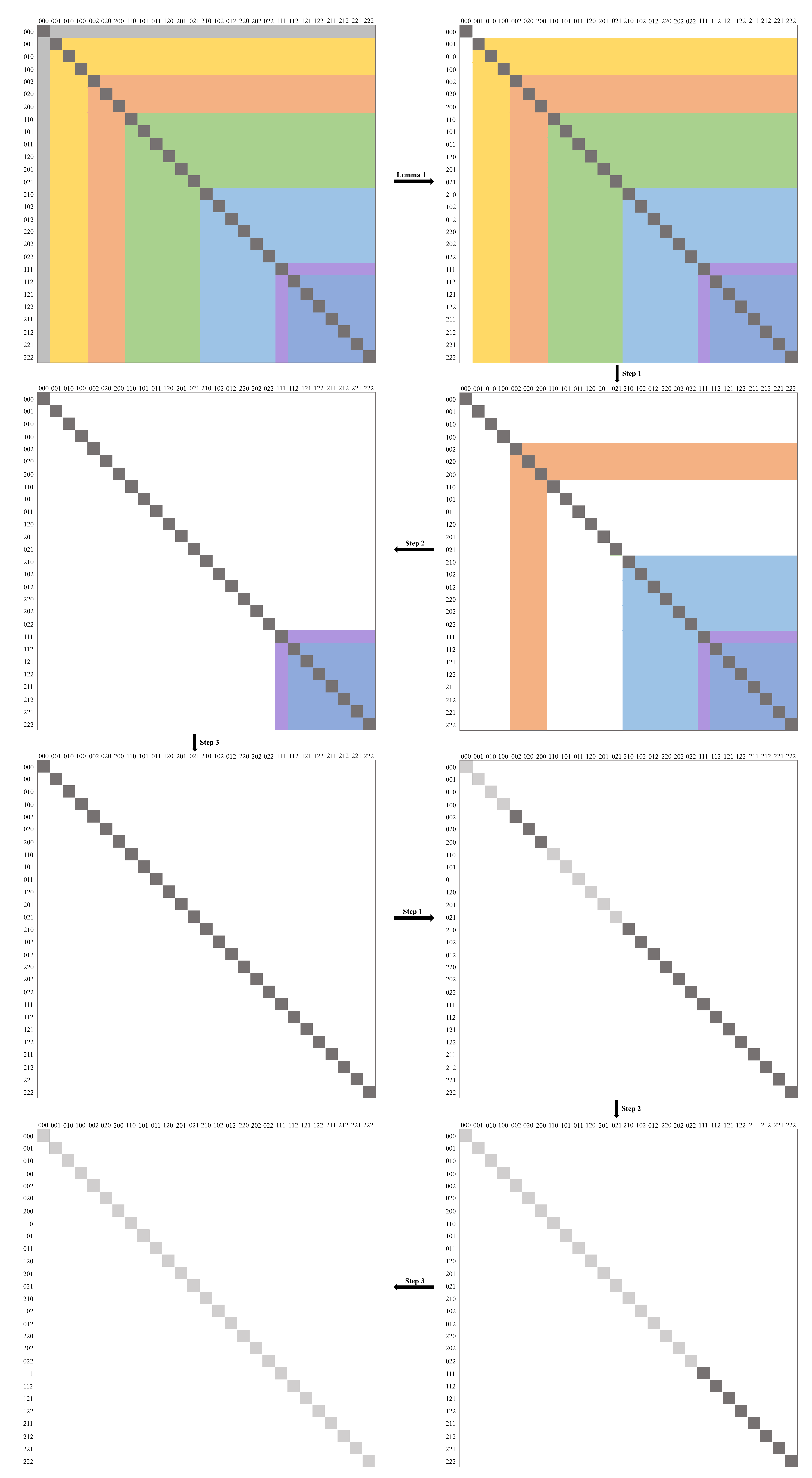}
	\caption{The proof details of strongest nonlocal set in $\mathbb{C}^{3}\otimes \mathbb{C}^{3}\otimes \mathbb{C}^{3}\otimes \mathbb{C}^{3}$. }
    \label{fig3}
\end{figure}
    Thus $E\propto I$. The proof details is described in Fig. \ref{fig3}. This completes the proof.
\end{proof}

Now, we give the strongest nonlocal sets in $\mathbb{C}^{d}\otimes \mathbb{C}^{d}\otimes \mathbb{C}^{d}\otimes \mathbb{C}^{d}$. Define
    \begin{equation*}\small
      \begin{aligned}
      &\mathcal{B}_{1}=\{ |\alpha_1\rangle = |1000\rangle, |\alpha_2\rangle = |0001\rangle,\\
      &~~~~~~~~~ |\alpha_3\rangle = |0010\rangle, |\alpha_4\rangle = |0100\rangle\},\\
      &\mathcal{B}_{2}=\{ \frac{1}{\sqrt{2} } (|0ji0\rangle +|i00j\rangle)~|~ 1\le i,j \le d-1\},\\
      &\mathcal{B}_{3}=\{ \frac{1}{\sqrt{2} } (|0i0j\rangle +|i0j0\rangle)~|~ 1\le i,j \le d-1\},\\
      &\mathcal{B}_{4}=\{ \frac{1}{\sqrt{2} } (|00ij\rangle +|ij00\rangle)~|~ 1\le i,j \le d-1\},\\
      &\mathcal{B}_{5}=\{ \frac{1}{\sqrt{2} } (|000i\rangle +|(i-1) ~ (i-1) ~ (i-1) ~ 0\rangle)~|\\
      &~~~~~~~~~2\le i \le d-1\},\\
      &\mathcal{B}_{6}=\{ \frac{1}{\sqrt{2} } (|00i0\rangle +|(i-1) ~ (i-1) ~ 0 ~ (i-1)\rangle)~|\\
      &~~~~~~~~~2\le i \le d-1\},\\
      \end{aligned}
    \end{equation*}
    \begin{equation*}\small
        \begin{aligned}
      &\mathcal{B}_{7}=\{ \frac{1}{\sqrt{2} } (|0i00\rangle +|(i-1) ~ 0 ~ (i-1) ~ (i-1)\rangle)~|\\
      &~~~~~~~~~2\le i \le d-1\},\\
      &\mathcal{B}_{8}=\{ \frac{1}{\sqrt{2} } (|0~(i-1) ~ (i-1) ~ (i-1)\rangle +|i000\rangle)~|\\
      &~~~~~~~~~2\le i \le d-1\},\\
      &\mathcal{B}_{9}=\{ \frac{1}{2} (|0ijk\rangle +|ijk0\rangle +|jk0i\rangle +|k0ij\rangle) ~|\\
      &~~~~~~~~~1\le i,j,k \le d-1 \backslash (i,j,k)=(1,1,1), \cdots,\\
      &~~~~~~~~~(d-2, d-2, d-2)\}.
      \end{aligned}
    \end{equation*}

\begin{theorem}\label{th:dddd}
    In $\mathbb{C}^{d}\otimes \mathbb{C}^{d}\otimes \mathbb{C}^{d}\otimes \mathbb{C}^{d}$ $(d\ge 2)$, the set $\{|\psi_{i} \rangle\}_{i=0}^{d^3}$ below is a strongest nonlocal set of size $d^3+1$,
    \begin{equation}\label{eq:dddd}
        \begin{aligned}
            &|\psi_{0} \rangle= |\alpha_{0} \rangle, \\
            &|\psi_{i} \rangle= \sum_{j=1}^{d^3} \omega_{d^3}^{ij}~|\alpha_{j} \rangle, ~~~ 1\le i \le d^3,
        \end{aligned}
    \end{equation}
    where
    \begin{equation}
        \begin{aligned}
            &|\alpha_0\rangle= |0000\rangle, \\
            &\{|\alpha_j\rangle \}_{j=1}^{d^3}= \mathcal{B}_{1}\cup \mathcal{B}_{2}\cup \mathcal{B}_{3}\cup \cdots \cup \mathcal{B}_{9}.
        \end{aligned}
    \end{equation}
\end{theorem}
 The proof of Theorem \ref{th:dddd} is given in Appendix \ref{append:B}. \\

Furthermore, we obtain the strongest nonlocal sets in general four-partite systems $\mathbb{C}^{d_{1}}\otimes \mathbb{C}^{d_{2}}\otimes \mathbb{C}^{d_{3}}\otimes \mathbb{C}^{d_{4}}$ $(2\leq d_{1}\leq d_{2}\leq d_{3}\leq d_{4})$. Define
    \begin{equation*}\small
        \begin{aligned}
        &\mathcal{B}_{1}=\{ |\alpha_1\rangle = |1000\rangle, |\alpha_2\rangle = |0001\rangle, \\
        &~~~~~~~~~|\alpha_3\rangle = |0010\rangle, |\alpha_4\rangle = |0100\rangle\},\\
        &\mathcal{B}_{2}=\{ \frac{1}{\sqrt{2} } (|0ji0\rangle +|i00j\rangle)~|~1\le i \le d_1-1,\\
        &~~~~~~~~~1 \le j \le d_2-1\}, \\
        &\mathcal{B}_{3}=\{ |0ji0\rangle~|~d_1\le i \le d_2-1,~ 1 \le j \le d_2-1\},\\
        &\mathcal{B}_{4}=\{ |0ji0\rangle~|~d_2\le i \le d_3-1,~ d_1 \le j \le d_2-1\},\\
        &\mathcal{B}_{5}=\{ \frac{1}{\sqrt{2} } (|0ji0\rangle +|j00i\rangle)~|~d_2\le i \le d_3-1,\\
        &~~~~~~~~~1 \le j \le d_1-1\},\\
        &\mathcal{B}_{6}=\{ \frac{1}{\sqrt{2} } (|0i0j\rangle +|i0j0\rangle)~|~1\le i \le d_1-1,\\
        &~~~~~~~~~1 \le j \le d_3-1\},\\
        &\mathcal{B}_{7}=\{ |0i0j\rangle~|~1\le i \le d_1-1,~d_3 \le j \le d_4-1\},\\
        &\mathcal{B}_{8}=\{ |0i0j\rangle~|~d_1\le i \le d_2-1,~1 \le j \le d_4-1\},\\
        &\mathcal{B}_{9}=\{ \frac{1}{\sqrt{2} } (|00ij\rangle +|ij00\rangle)~|~1\le i \le d_1-1,\\
        &~~~~~~~~~ 1 \le j \le d_2-1\},\\
        &\mathcal{B}_{10}=\{ |00ij\rangle~|~1\le i \le d_1-1,~d_2 \le j \le d_4-1\},\\
        \end{aligned}
    \end{equation*}
    \begin{equation*}\small
        \begin{aligned}
        &\mathcal{B}_{11}=\{ |00ij\rangle~|~d_1\le i \le d_3-1,~1 \le j \le d_4-1\},\\
        &\mathcal{B}_{12}=\{ \frac{1}{\sqrt{2} } (|000i\rangle +|(i-1)~(i-1)~(i-1)~0\rangle)~|\\
        &~~~~~~~~~2\le i \le d_1\},\\
        &\mathcal{B}_{13}=\{ \frac{1}{\sqrt{2} } (|00i0\rangle +|(i-1)~(i-1)~0~(i-1)\rangle)~|\\
        &~~~~~~~~~2\le i \le d_1\},\\
        &\mathcal{B}_{14}=\{ \frac{1}{\sqrt{2} } (|0i00\rangle +|(i-1)~0~(i-1)~(i-1)\rangle)~|\\
        &~~~~~~~~~2\le i \le d_1\},\\
        &\mathcal{B}_{15}=\{ |000i\rangle ~|~d_1+1\le i \le d_4-1\},\\
        &\mathcal{B}_{16}=\{ |00i0\rangle ~|~d_1+1\le i \le d_3-1\},\\
        &\mathcal{B}_{17}=\{ |0i00\rangle ~|~d_1+1\le i \le d_2-1\},\\
        &\mathcal{B}_{18}=\{ \frac{1}{\sqrt{2} } (|0~(i-1)~(i-1)~(i-1)\rangle +|i000\rangle)~|\\
        &~~~~~~~~~2\le i \le d_1-1\},\\
        &\mathcal{B}_{19}=\{ |0~(i-1)~(i-1)~(i-1)\rangle~|~d_1\le i \le d_2\},\\
        &\mathcal{B}_{20}=\{ \frac{1}{2} (|0ijk\rangle +|ijk0\rangle +|jk0i\rangle +|k0ij\rangle)~|\\
        &~~~~~~~~~1\le i,j,k \le d_1-1\\
        &~~~~~~~~ \backslash (i,j,k)=(1,1,1), \cdots ,(d_1-1, d_1-1, d_1-1)\},\\
        &\mathcal{B}_{21}=\{ \frac{1}{\sqrt{3}} (|0ijk\rangle +|ijk0\rangle +|jk0i\rangle)~|~1\le i \le d_1-1,\\
        &~~~~~~~~~1 \le j \le d_1-1,~ d_1 \le k \le d_2-1\},\\
        &\mathcal{B}_{22}=\{ \frac{1}{\sqrt{3}} (|0ijk\rangle +|ijk0\rangle +|k0ij\rangle)~|~1\le i \le d_1-1,\\
        &~~~~~~~~~d_1 \le j \le d_2-1,~1 \le k \le d_1-1\},\\
        &\mathcal{B}_{23}=\{ \frac{1}{\sqrt{3}} (|0ijk\rangle +|jk0i\rangle +|k0ij\rangle)~|~d_1 \le i \le d_2-1,\\
        &~~~~~~~~~1 \le j \le d_1-1,~1 \le k \le d_1-1\},\\
        &\mathcal{B}_{24}=\{ \frac{1}{\sqrt{3}} (|0ijk\rangle +|ijk0\rangle +|i0kj\rangle)~|~1\le i \le d_1-1,\\
        &~~~~~~~~~1 \le j \le d_1-1,~ d_2 \le k \le d_3-1\},\\
        &\mathcal{B}_{25}=\{ \frac{1}{\sqrt{3}} (|0ijk\rangle +|k0ij\rangle +|ik0j\rangle)~|~1\le i \le d_1-1,\\
        &~~~~~~~~~d_2 \le j \le d_3-1,~1 \le k \le d_1-1\},\\
        &\mathcal{B}_{26}=\{ \frac{1}{\sqrt{3}} (|0ijk\rangle +|j0ik\rangle +|ij0k\rangle)~|~1\le i \le d_1-1, \\
        &~~~~~~~~~1 \le j \le d_1-1,~ d_3 \le k \le d_4-1\},\\
        &\mathcal{B}_{27}=\{ \frac{1}{\sqrt{3}} (|0ijk\rangle +|i0jk\rangle +|ij0k\rangle)~|~1\le i \le d_1-1, \\
        &~~~~~~~~~d_1 \le j \le d_2-1,~ d_3 \le k \le d_4-1\},\\
        &\mathcal{B}_{28}=\{ \frac{1}{\sqrt{2}} (|0ijk\rangle +|ijk0\rangle )~|~1\le i \le d_1-1, \\
        &~~~~~~~~~d_1 \le j \le d_2-1,~ d_1 \le k \le d_3-1\},\\
        &\mathcal{B}_{29}=\{ \frac{1}{\sqrt{2}} (|0ijk\rangle +|k0ij\rangle)~|~d_1 \le i \le d_2-1, \\
        &~~~~~~~~~d_1 \le j \le d_3-1,~1 \le k \le d_1-1 \},\\
        &\mathcal{B}_{30}=\{ \frac{1}{\sqrt{2}} (|0ijk\rangle +|jk0i\rangle)~|~d_1 \le i \le d_2-1, \\
        &~~~~~~~~~1 \le j \le d_1-1,~d_1 \le k \le d_2-1\},\\
        \end{aligned}
    \end{equation*}
    \begin{equation*}\small
        \begin{aligned}
        &\mathcal{B}_{31}=\{ \frac{1}{\sqrt{2}} (|0ijk\rangle +|i0jk\rangle)~|~1 \le i \le d_1-1, \\
        &~~~~~~~~~d_2 \le j \le d_3-1,~d_1 \le k \le d_4-1\},\\
        &\mathcal{B}_{32}=\{ \frac{1}{\sqrt{2}} (|0ijk\rangle +|ji0k\rangle)~|~d_1 \le i \le d_2-1, \\
        &~~~~~~~~~1 \le j \le d_1-1, ~d_2 \le k \le d_3-1\},\\
        &\mathcal{B}_{33}=\{ \frac{1}{\sqrt{2}} (|0~ (d_2-1) ~ ij\rangle +|i00j\rangle)~|~ 1 \le i \le d_1-1,\\
        &~~~~~~~~ ~d_3 \le j \le d_4-1\},\\
        &\mathcal{B}_{34}=\{ |0ijk\rangle ~|~d_1 \le i \le d_2-1,~d_2 \le j \le d_3-1, \\
        &~~~~~~~~ ~d_1 \le k \le d_4-1\},\\
        &\mathcal{B}_{35}=\{ |0ijk\rangle ~|~d_1 \le i \le d_2-2,~1 \le j \le d_1-1, \\
        &~~~~~~~~ ~d_3 \le k \le d_4-1\},\\
        &\mathcal{B}_{36}=\{ |0ijk\rangle ~|~d_1\le i,j,k \le d_2-1 \backslash (i,j,k)=(d_1, d_1, d_1),\\
        &~~~~~~~~~\cdots ,(d_2-1, d_2-1, d_2-1)\}.
        \end{aligned}
    \end{equation*}

\begin{theorem}\label{th:d1d2d3d4}
    In $\mathbb{C}^{d_{1}}\otimes \mathbb{C}^{d_{2}}\otimes \mathbb{C}^{d_{3}}\otimes \mathbb{C}^{d_{4}}$~$(2\le d_1 \le d_2 \le d_3 \le d_4)$, the set $\{|\psi_{i} \rangle\}_{i=0}^{d_2d_3d_4}$ below is a strongest nonlocal set of size $d_2d_3d_4+1$,
    \begin{equation}\label{eq:d1d2d3d4}
        \begin{aligned}
            |\psi_{0} \rangle &= |\alpha_{0} \rangle, \\
            |\psi_{i} \rangle &= \sum_{j=1}^{d_2d_3d_4} \omega_{d_2d_3d_4}^{ij}~|\alpha_{j} \rangle, ~~~ 1\le i \le d_2d_3d_4.
        \end{aligned}
    \end{equation}
    where
    \begin{equation}
        \begin{aligned}
            &|\alpha_0\rangle= |0000\rangle, \\
            &\{|\alpha_j\rangle \}_{j=1}^{d_2d_3d_4}= \mathcal{B}_{1}\cup \mathcal{B}_{2}\cup \mathcal{B}_{3}\cup \cdots \cup \mathcal{B}_{36}.
        \end{aligned}
    \end{equation}
\end{theorem}

The proof of Theorem \ref{th:d1d2d3d4} is given in Appendix \ref{append:C}. \\

\section{Conclusion}\label{sec5}
Quantum nonlocality based on local indistinguishability has been studied extensively. The strongest nonlocal set is the strongest form of nonlocality, which is widely applied in multipartite quantum secret sharing and other cryptographic protocols. In order to improve communication efficiency and save resources, it is of great importance to construct the strongest nonlocal sets whose number of the elements reaches the lower bound for multipartite systems.

In this paper, we found the strongest nonlocal sets with the minimum cardinality in $\mathbb{C}^{d_{1}}\otimes \mathbb{C}^{d_{2}}\otimes \mathbb{C}^{d_{3}}$ and $\mathbb{C}^{d_{1}}\otimes \mathbb{C}^{d_{2}}\otimes \mathbb{C}^{d_{3}}\otimes \mathbb{C}^{d_{4}}$, which further proves the conjecture proposed in Ref. \cite{Li2023}. For clarity, we compare our results with the strongest nonlocal set available (see Table \ref{tab1}). Specifically, each set we constructed is given by genuinely entangled states, except for one product state which is not a stopper state.  Our results may highlight studies on the strongest nonlocal sets with the minimum cardinality in $\otimes _{i=1}^{N}\mathbb{C}^{d_i}$ for $N\geq5$ and the gap between strong nonlocality and strongest nonlocality.

\begin{table*}[htbp]
		\newcommand{\tabincell}[2]{\begin{tabular}{@{}#1@{}}#2\end{tabular}}
		\centering
		\caption{\label{tab1}Results about the strongest nonlocal sets in tripartite and four-partite systems. }
\scalebox{0.86}{
\begin{tabular}{cccc}
			\toprule
			\hline
			\hline
			\specialrule{0em}{1.5pt}{1.5pt}
			~Type~~&~~Systems~~&~~Cardinality~~&~~~References~~\\
			\specialrule{0em}{1.5pt}{1.5pt}
			\midrule
			\hline
			\specialrule{0em}{3.5pt}{3.5pt}
			\tabincell{c}{OGES and $|000\rangle$\\OGES and $|0000\rangle$}&\tabincell{c}{$\mathbb{C}^{d_1}\otimes \mathbb{C}^{d_2}\otimes \mathbb{C}^{d_3}$\\$\mathbb{C}^{d_1}\otimes \mathbb{C}^{d_2}\otimes \mathbb{C}^{d_3}\otimes \mathbb{C}^{d_4}$}&\tabincell{c}{$d_1d_2d_3-(d_1-1)(d_2-1)(d_3-1)$\\~~$d_1d_2d_3d_4-(d_1-1)(d_2-1)(d_3-1)(d_4-1)$}&~~\tabincell{c}{\cite{Li2023} }\\
			\specialrule{0em}{3.5pt}{3.5pt}
			\tabincell{c}{OGES}&\tabincell{c}{$\mathbb{C}^{3}\otimes \mathbb{C}^{3}\otimes \mathbb{C}^{3}$\\$\mathbb{C}^{3}\otimes \mathbb{C}^{3}\otimes \mathbb{C}^{3}\otimes \mathbb{C}^{3}$}&\tabincell{c}{$2\times 3^2$\\$2\times 3^3$}&~~\tabincell{c}{\cite{Hu2024} }\\
			\specialrule{0em}{3.5pt}{3.5pt}
            \tabincell{c}{OGES and $|000\rangle$\\OGES and $|0000\rangle$}&\tabincell{c}{$\mathbb{C}^{d_1}\otimes \mathbb{C}^{d_2}\otimes \mathbb{C}^{d_3}$\\$\mathbb{C}^{d}\otimes \mathbb{C}^{d}\otimes \mathbb{C}^{d}\otimes \mathbb{C}^{d}$}&\tabincell{c}{$d_2d_3+d_1-1$\\$d^3+d-1$}&~~\tabincell{c}{\cite{Shi2023} }\\
			\specialrule{0em}{3.5pt}{3.5pt}		
			\tabincell{c}{OGES and stopper states}&\tabincell{c}{$\mathbb{C}^{d_1}\otimes \mathbb{C}^{d_2}\otimes \mathbb{C}^{d_3}$}&\tabincell{c}{$d_2d_3+1$~(lower bound)}&~~\tabincell{c}{\cite{Zhen2024} }\\	
			\specialrule{0em}{3.5pt}{3.5pt}
			\tabincell{c}{OGES and $|000\rangle$}&\tabincell{c}{$\mathbb{C}^{d_1}\otimes \mathbb{C}^{d_2}\otimes \mathbb{C}^{d_3}$}&\tabincell{c}{$d_2d_3+1$~(lower bound)}&~~\tabincell{c}{Theorem~\ref{th:d1d2d3}}\\
			\specialrule{0em}{3.5pt}{3.5pt}
			\tabincell{c}{OGES and $|0000\rangle$}&\tabincell{c}{$\mathbb{C}^{d}\otimes \mathbb{C}^{d}\otimes \mathbb{C}^{d}\otimes \mathbb{C}^{d}$\\$\mathbb{C}^{d_1}\otimes \mathbb{C}^{d_2}\otimes \mathbb{C}^{d_3}\otimes \mathbb{C}^{d_4}$}&\tabincell{c}{$d^3+1$~(lower bound)\\$d_2d_3d_4+1$~(lower bound)}&~~\tabincell{c}{Theorem~\ref{th:dddd}\\Theorem~\ref{th:d1d2d3d4}}\\		
			\bottomrule
			\specialrule{0em}{1.5pt}{1.5pt}
			\hline
			\hline
		\end{tabular}}
	\end{table*}

\section*{Acknowledgments}
This work is supported by Basic Research Project of Shijiazhuang Municipal Universities in Hebei Province (241790697A), Natural Science Foundation of Hebei Province (F2021205001), NSFC (Grant Nos. 62272208, 12171044 and 12075159), the specific research fund of the Innovation Platform for Academicians of Hainan Province under Grant No. YSPTZX202215.
arxivpdfoutput

\bibliographystyle{quantum}
\bibliography{main}

\onecolumn
\appendix
\section{Proof of Theorem \ref{th:d1d2d3}}\label{append:A}
\begin{proof}
  Suppose that B and C perform a joint OPLM $\{E=M^{\dagger}M\}$. The post-measurement states should be mutually orthogonal, i.e., $\langle \psi_{i}| I\otimes E |\psi_{j} \rangle = 0$ for $0\le i\ne j \le d_2d_3$.

 Note that $\{|\psi_{0} \rangle\}$ is spanned by $\{|\alpha_{0}\rangle\}$ and $\{|\psi_{i} \rangle\}_{i=1}^{d_2d_3}$ is spanned by $\{|\alpha_{j} \rangle\}_{j=1}^{d_2d_3}$. As $\langle \psi_{0}| I\otimes E |\psi_{i} \rangle = 0$ for $1 \le i \le d_2d_3$, applying Lemma \ref{lem:zero} to $\{|\psi_{0}\rangle\}$ and $\{|\psi_{i} \rangle\}_{i=1}^{d_2d_3}$, we obtain that $\langle \alpha_{0}| I\otimes E |\alpha_{j} \rangle = 0$ for $1\le j \le d_2d_3$, which implies that
  \begin{equation*}
        \begin{aligned}
            \langle 00|E|mn\rangle =\langle mn|E|00\rangle = 0,
        \end{aligned}
    \end{equation*}
  where $(m,n)\in \mathbb{Z}_{d_2} \times \mathbb{Z}_{d_3}\backslash (0,0)$. Hence, we have
    \begin{equation*}
        \begin{aligned}
            &\langle \alpha_{1}| I\otimes E |\alpha_{j} \rangle = 0, ~~2\le j\le d_2d_3,\\
            &\langle \alpha_1 |\psi_j \rangle \ne 0, ~~1\le j \le d_2d_3.
        \end{aligned}
    \end{equation*}
 Applying Lemma \ref{lem:trivial} to $\{|\psi_{i} \rangle\}_{i=1}^{d_2d_3}$, we obtain
\begin{equation*}
       \begin{aligned}
            &\langle \alpha_{x}| I\otimes E |\alpha_{y} \rangle = 0, ~~1\le x\ne y\le d_2d_3,\\
            &\langle \alpha_{x}| I\otimes E |\alpha_{x} \rangle = \langle \alpha_{y}| I\otimes E |\alpha_{y} \rangle, ~~1\le x\ne y \le d_2d_3.
       \end{aligned}
    \end{equation*}

    First, we prove that the off-diagonal entries of the matrix $E$ are zeros.\\
    \textbf{Step 1}
    From $\langle001| I\otimes E |\alpha_y\rangle =\langle010| I\otimes E |\alpha_y\rangle = 0$, we have
    \begin{equation*}
        \begin{aligned}
            \langle01| E |mn\rangle =\langle10| E |mn\rangle = 0.
        \end{aligned}
    \end{equation*}
    \textbf{Step \bm{$k~(2\le k \le d_1)$}}
    From the ($k-1$)-th step, we get $\langle 0~(k-1)| E |mn\rangle =\langle (k-1)~0| E |mn\rangle =0$. Then $\langle \alpha_{x}| I\otimes E |\alpha_{y} \rangle = 0$ where $|\alpha_{x} \rangle \in \mathcal{B}_{2,3}$ with $i=k$ implies that
     \begin{equation*}
       \begin{aligned}
            \langle 0k| E |mn\rangle =\langle k0| E |mn\rangle=0,~2\le k \le d_1.
       \end{aligned}
    \end{equation*}
    \textbf{Step \bm{$d_1+1$}}
    For $|\alpha_{x} \rangle \in \mathcal{B}_{4,5}$, we obtain
    \begin{equation*}
       \begin{aligned}
           \langle 0i| E |mn\rangle =0,~d_1+1\le i \le d_3-1; \\
           \langle i0| E |mn\rangle = 0,~d_1+1\le i \le d_2-1.
       \end{aligned}
    \end{equation*}
    \textbf{Step \bm{$d_1+2$}}
    According to the above results, for $|\alpha_{x} \rangle \in \mathcal{B}_{6-13}$, we have
    \begin{equation*}
       \begin{aligned}
            \langle ij| E |mn\rangle = 0,~1\le i\le d_2-1, 1\le j\le d_3-1.
       \end{aligned}
    \end{equation*}

    Next, we prove that all diagonal entries of the matrix $E$ are equal.\\
    \textbf{Step 1}
    From $\langle100| I\otimes E |100\rangle =\langle001| I\otimes E |001\rangle =\langle010| I\otimes E |010\rangle $, we have
    \begin{equation*}
       \begin{aligned}
            \langle00| E |00\rangle =\langle01| E |01\rangle=\langle10| E |10\rangle.
       \end{aligned}
    \end{equation*}
    \textbf{Step \bm{$k~(2\le k \le d_1)$}}
    From the ($k-1$)-th step, we get $\langle00| E |00\rangle =\langle0~(k-1)| E |0~(k-1)\rangle=\langle (k-1)~0| E |(k-1)~0\rangle$. Then $\langle100| I\otimes E |100\rangle =\langle \alpha_{y}| I\otimes E |\alpha_{y} \rangle$ where $|\alpha_{y} \rangle \in \mathcal{B}_{2,3}$ with $i=k$ implies that
    \begin{equation*}
       \begin{aligned}
            \langle00| E |00\rangle =\langle 0k| E |0k\rangle=\langle k0| E |k0\rangle,~2\le k \le d_1.
       \end{aligned}
    \end{equation*}
    \textbf{Step \bm{$d_1+1$}}
    For $|\alpha_{y} \rangle \in \mathcal{B}_{4,5}$, we obtain
    \begin{equation*}
       \begin{aligned}
            \langle00| E |00\rangle=\langle0i| E |0i\rangle,~d_1+1\le i \le d_3-1;\\
            \langle00| E |00\rangle=\langle i0| E |i0\rangle,~d_1+1\le i \le d_2-1.
       \end{aligned}
    \end{equation*}
    \textbf{Step \bm{$d_1+2$}}
    Based on the above results, for $|\alpha_{x} \rangle \in \mathcal{B}_{6-13}$, we have
    \begin{equation*}
       \begin{aligned}
            \langle00| E |00\rangle =\langle ij| E |ij\rangle,~1\le i\le d_2-1, 1\le j\le d_3-1.
       \end{aligned}
    \end{equation*}

    Thus $E\propto I$. This completes the proof.
\end{proof}

\section{Proof of Theorem \ref{th:dddd}}\label{append:B}
\begin{proof}
     Suppose that $B$, $C$, and $D$ perform a joint OPLM $\{E=M^{\dagger}M\}$. The post-measurement states should be mutually orthogonal, i.e., $\langle \psi_{i}| I\otimes E |\psi_{j} \rangle = 0$ for $0\le i\ne j \le d^3$.

    Note that $\{|\psi_{0}\rangle\}$ is spanned by $\{|\alpha_{0}\rangle\}$ and $\{|\psi_{i} \rangle\}_{i=1}^{d^3}$ is spanned by $\{|\alpha_{j} \rangle\}_{j=1}^{d^3}$. As $\langle \psi_{0}| I\otimes E |\psi_{j} \rangle = 0$ for $0\le i\ne j \le d^3$, applying Lemma \ref{lem:zero} to $\{|\psi_{0}\rangle\}$ and $\{|\psi_{i} \rangle\}_{i=1}^{d^3}$, we obtain $\langle \alpha_{0}| I\otimes E |\alpha_{j} \rangle = 0$ for $1\le j \le d^3$, which implies that
    \begin{equation*}
        \begin{aligned}
            \langle 000|E|mnl\rangle =\langle mnl|E|000\rangle =0,
        \end{aligned}
    \end{equation*}
    where $(m,n,l)\in \mathbb{Z}_3 \times \mathbb{Z}_3 \times \mathbb{Z}_3\backslash (0,0,0)$. Therefore,
    \begin{equation*}
        \begin{aligned}
            &\langle \alpha_{1}| I\otimes E |\alpha_{j} \rangle = 0, ~~2\le j\le d^3,\\
            &\langle \alpha_1 |\psi_j \rangle \ne 0, ~~1\le j \le d^3.
        \end{aligned}
    \end{equation*}
    Applying Lemma \ref{lem:trivial} to $\{|\psi_{i} \rangle\}_{i=1}^{d^3}$, we have
    \begin{equation*}
        \begin{aligned}
            \langle \alpha_{x}| I\otimes E |\alpha_{y} \rangle &= 0, ~~1\le x\ne y\le d^3,\\
            \langle \alpha_{x}| I\otimes E |\alpha_{x} \rangle &= \langle \alpha_{y}| I\otimes E |\alpha_{y} \rangle, ~~1\le x\ne y \le d^3.
        \end{aligned}
    \end{equation*}

    First, we proof that the off-diagonal entries of $E$ are zeros.\\
    \textbf{Step 1}
    From $\langle \alpha_{x}| I\otimes E |\alpha_{y} \rangle = 0$ for $2 \le x \le 4$, we have
    \begin{equation*}
        \begin{aligned}
            \langle001| E |mnl\rangle =\langle010| E |mnl\rangle =\langle100| E |mnl\rangle = 0.
        \end{aligned}
    \end{equation*}
    Then $\langle \alpha_{x}| I\otimes E |\alpha_{y} \rangle = 0$ where $|\alpha_{x} \rangle \in \mathcal{B}_{2,3,4}$ with $j=1$ implies that
    \begin{equation*}
        \begin{aligned}
            \langle 1i0| E |mnl\rangle =\langle i01| E |mnl\rangle =\langle 0i1| E |mnl\rangle = 0,~1 \le i \le d-1.
        \end{aligned}
    \end{equation*}
    \textbf{Step \bm{$k~(2\le k \le d-1)$}}
    For $|\alpha_{x} \rangle \in \mathcal{B}_{5,6,7}$ with $i=k$, we obtain
    \begin{equation*}
        \begin{aligned}
            \langle 00k| E |mnl\rangle =\langle 0k0| E |mnl\rangle =\langle k00| E |mnl\rangle = 0,~2\le k \le d-1.
        \end{aligned}
    \end{equation*}
    Then $\langle \alpha_{x}| I\otimes E |\alpha_{y} \rangle = 0$ where $|\alpha_{x} \rangle \in \mathcal{B}_{2,3,4}$ with $j=k$ implies that
    \begin{equation*}
        \begin{aligned}
            \langle ki0| E |mnl\rangle =\langle i0k| E |mnl\rangle =\langle 0ik| E |mnl\rangle = 0,~1 \le i \le d-1, ~2\le k \le d-1.
        \end{aligned}
    \end{equation*}
    \textbf{Step \bm{$d$}}
    According to the above results, with respect to $|\alpha_{x} \rangle \in \mathcal{B}_{8,9}$, we have
    \begin{equation*}
        \begin{aligned}
            \langle ijk| E |mnl\rangle = 0,~1\le i,j,k \le d-1.
        \end{aligned}
    \end{equation*}

    Next, we prove that all the diagonal entries $E$ are equal.\\
    \textbf{Step 1}
    From $\langle \alpha_{1}| I\otimes E |\alpha_{1} \rangle = \langle \alpha_{y}| I\otimes E |\alpha_{y} \rangle$ for $2\le y \le 4$, we have
    \begin{equation*}
        \begin{aligned}
            \langle000| E |000\rangle = \langle001| E |001\rangle = \langle010| E |010\rangle = \langle100| E |100\rangle.
        \end{aligned}
    \end{equation*}
    Then $\langle \alpha_{1}| I\otimes E |\alpha_{1} \rangle = \langle \alpha_{y}| I\otimes E |\alpha_{y} \rangle$ where $|\alpha_{y} \rangle \in \mathcal{B}_{2,3,4}$ with $j=1$ implies that
    \begin{equation*}
        \begin{aligned}
            \langle000| E |000\rangle = \langle1i0| E |1i0\rangle = \langle i01| E |i01\rangle = \langle0i1| E |0i1\rangle,~1\le i \le d-1.
        \end{aligned}
    \end{equation*}
    \textbf{Step \bm{$k~(2\le k \le d-1)$}}
   For $|\alpha_{y} \rangle \in \mathcal{B}_{5,6,7}$ with $i=k$, we obtain
    \begin{equation*}
        \begin{aligned}
            \langle000| E |000\rangle = \langle00k| E |00k\rangle = \langle0k0| E |0k0\rangle = \langle k00| E |k00\rangle,~2\le k \le d-1.
        \end{aligned}
    \end{equation*}
    Then considering $|\alpha_{y} \rangle \in \mathcal{B}_{2,3,4}$ with $j=k$, we have
    \begin{equation*}
        \begin{aligned}
            \langle000| E |000\rangle = \langle ki0| E |ki0\rangle = \langle i0k| E |i0k\rangle = \langle0ik| E |0ik\rangle,~1 \le i \le d-1, ~2\le k \le d-1.
        \end{aligned}
    \end{equation*}
    \textbf{Step \bm{$d$}}
    Based on the above results, for $|\alpha_{y} \rangle \in \mathcal{B}_{8,9}$, we have
    \begin{equation*}
        \begin{aligned}
            \langle000| E |000\rangle = \langle ijk| E |ijk\rangle,~1\le i,j,k \le d-1.
        \end{aligned}
    \end{equation*}

    Thus $E\propto I$. This completes the proof.
\end{proof}

\section{Proof of Theorem \ref{th:d1d2d3d4}}\label{append:C}
\begin{proof}
    Suppose that $B$, $C$, and $D$ perform a joint OPLM $\{E=M^{\dagger}M\}$. The post-measurement states should be mutually orthogonal, i.e., $\langle \psi_{i}| I\otimes E |\psi_{j} \rangle = 0$ for $0\le i\ne j \le d_2d_3d_4$.

    Note that $\{|\psi_{0}\rangle\}$ is spanned by $\{|\alpha_{0}\rangle\}$ and $\{|\psi_{i} \rangle\}_{i=1}^{d_2d_3d_4}$ is spanned by $\{|\alpha_{j} \rangle\}_{j=1}^{d_2d_3d_4}$. As $\langle \psi_{0}| I\otimes E |\psi_{j} \rangle = 0$ for $0\le j \le d_2d_3d_4$, applying Lemma \ref{lem:zero} to the set $\{|\psi_{0}\rangle\}$ and $\{|\psi_{i} \rangle\}_{i=1}^{d_2d_3d_4}$, we obtain $\langle \alpha_{0}| I\otimes E |\alpha_{j} \rangle = 0$ for $1\le j \le d_2d_3d_4$, which implies that
    \begin{equation*}
        \begin{aligned}
            \langle 000|E|mnl\rangle =\langle mnl|E|000\rangle = 0,
        \end{aligned}
    \end{equation*}
    where $(m,n,l)\in \mathbb{Z}_3 \times \mathbb{Z}_3 \times \mathbb{Z}_3\backslash (0,0,0)$. Therefore,
    \begin{equation*}
        \begin{aligned}
            &\langle \alpha_{1}| I\otimes E |\alpha_{j} \rangle= 0, ~~2\le j\le d_2d_3d_4,\\
            &\langle \alpha_1 |\psi_j \rangle\ne 0, ~~1\le j \le d_2d_3d_4.
        \end{aligned}
    \end{equation*}
    Applying Lemma \ref{lem:trivial} to $\{|\psi_{i} \rangle\}_{i=1}^{d_2d_3d_4}$, we have
    \begin{equation*}
        \begin{aligned}
            &\langle \alpha_{x}| I\otimes E |\alpha_{y} \rangle= 0, ~~1\le x\ne y\le d_2d_3d_4,\\
            &\langle \alpha_{x}| I\otimes E |\alpha_{x} \rangle= \langle \alpha_{y}| I\otimes E |\alpha_{y} \rangle, ~~1\le x\ne y \le d_2d_3d_4.
        \end{aligned}
    \end{equation*}

    First, we prove that the off-diagonal entries of $E$ are zeros.\\
    \textbf{Step 1}
    From $\langle \alpha_{x}| I\otimes E |\alpha_{y} \rangle = 0$ for $2 \le x \le 4$, we have
    \begin{equation*}
        \begin{aligned}
            \langle001| E |mnl\rangle =\langle010| E |mnl\rangle =\langle100| E |mnl\rangle = 0.
        \end{aligned}
    \end{equation*}
    Then $\langle \alpha_{x}| I\otimes E |\alpha_{y} \rangle = 0$ where $|\alpha_{x} \rangle \in \mathcal{B}_{2,6,9}$ with $j=1$ implies that
    \begin{equation*}
        \begin{aligned}
            \langle 1i0| E |mnl\rangle =\langle i01| E |mnl\rangle =\langle 0i1| E |mnl\rangle = 0,~1 \le i \le d_1-1.
        \end{aligned}
    \end{equation*}
    \textbf{Step \bm{$k~(2\le k \le d_1)$}}
    For $|\alpha_{x} \rangle \in \mathcal{B}_{12,13,14}$ with $i=k$, we obtain
    \begin{equation*}
        \begin{aligned}
            \langle 00k| E |mnl\rangle =\langle 0k0| E |mnl\rangle =\langle k00| E |mnl\rangle = 0,~2\le k \le d_1.
        \end{aligned}
    \end{equation*}
    Then $\langle \alpha_{x}| I\otimes E |\alpha_{y} \rangle = 0$ where $|\alpha_{x} \rangle \in \mathcal{B}_{2,6,9}$ with $j=k$ implies that
    \begin{equation*}
        \begin{aligned}
            \langle ki0| E |mnl\rangle =\langle i0k| E |mnl\rangle =\langle 0ik| E |mnl\rangle = 0,~1 \le i \le d_1-1, ~2\le k \le d_1.
        \end{aligned}
    \end{equation*}
    \textbf{Step \bm{$d_1+1$}}
    For $|\alpha_{x} \rangle \in \mathcal{B}_{15,16,17}$, we have
    \begin{equation*}
        \begin{aligned}
            \langle 00i| E |mnl\rangle = 0, ~d_1+1\le i \le d_4-1;\\
            \langle 0i0| E |mnl\rangle = 0, ~d_1+1\le i \le d_3-1;\\
            \langle i00| E |mnl\rangle = 0, ~d_1+1\le i \le d_2-1.
        \end{aligned}
    \end{equation*}
    Then $\langle \alpha_{x}| I\otimes E |\alpha_{y} \rangle = 0$ where $|\alpha_{x} \rangle \in \mathcal{B}_{2,6,9}$ with $j\ge d_1 +1$ means that
    \begin{equation*}
        \begin{aligned}
            \langle ji0| E |mnl\rangle = 0, ~1\le i \le d_1-1, d_1+1\le j \le d_2-1;\\
            \langle i0j| E |mnl\rangle = 0, ~1\le i \le d_1-1, d_1+1\le j \le d_3-1;\\
            \langle 0ij| E |mnl\rangle = 0, ~1\le i \le d_1-1, d_1+1\le j \le d_2-1,
        \end{aligned}
    \end{equation*}
    and $\langle \alpha_{x}| I\otimes E |\alpha_{y} \rangle = 0$ where $|\alpha_{x} \rangle \in \mathcal{B}_{5}$ means that
    \begin{equation*}
        \begin{aligned}
            \langle ji0| E |mnl\rangle = 0, ~d_2\le i \le d_3-1,~1\le j \le d_1-1.
        \end{aligned}
    \end{equation*}
    For $|\alpha_{x} \rangle \in \mathcal{B}_{3,4,7,8,10,11}$, we have
    \begin{equation*}
        \begin{aligned}
            \langle ji0| E |mnl\rangle = 0, ~d_1\le i \le d_2-1, ~~1\le j \le d_2-1;\\
            \langle ji0| E |mnl\rangle = 0, ~d_2\le i \le d_3-1, ~d_1\le j \le d_2-1;\\
            \langle i0j| E |mnl\rangle = 0, ~~1\le i \le d_1-1, ~d_3\le j \le d_4-1;\\
            \langle i0j| E |mnl\rangle = 0, ~d_1\le i \le d_2-1, ~~1\le j \le d_4-1;\\
            \langle 0ij| E |mnl\rangle = 0, ~~1\le i \le d_1-1, ~d_2\le j \le d_4-1;\\
            \langle 0ij| E |mnl\rangle = 0, ~d_1\le i \le d_3-1, ~~1\le j \le d_4-1.
        \end{aligned}
    \end{equation*}
    \textbf{Step \bm{$d_1+2$}}
    According to the above results, with respect to $|\alpha_{x} \rangle \in \mathcal{B}_{18-36}$, we have
    \begin{equation*}
        \begin{aligned}
            &\langle ijk| E |mnl\rangle = 0,~1 \le i \le d_2-1, ~1 \le j \le d_3-1, ~1 \le k \le d_4-1.
        \end{aligned}
    \end{equation*}

    Next, we prove that all the diagonal entries of $E$ are equal.\\
    \textbf{Step 1}
    From $\langle \alpha_{1}| I\otimes E |\alpha_{1} \rangle = \langle \alpha_{y}| I\otimes E |\alpha_{y} \rangle$ for $2 \le y \le 4$, we have
    \begin{equation*}
        \begin{aligned}
            \langle000| E |000\rangle = \langle001| E |001\rangle =\langle010| E |010\rangle =\langle100| E |100\rangle.
        \end{aligned}
    \end{equation*}
    Then $\langle \alpha_{1}| I\otimes E |\alpha_{1} \rangle = \langle \alpha_{y}| I\otimes E |\alpha_{y} \rangle$ where $|\alpha_{y} \rangle \in \mathcal{B}_{2,6,9}$ with $j=1$ implies that
    \begin{equation*}
        \begin{aligned}
            \langle000| E |000\rangle = \langle 1i0| E |1i0\rangle =\langle i01| E |i01\rangle =\langle 0i1| E |0i1\rangle,~1 \le i \le d_1-1.
        \end{aligned}
    \end{equation*}
    \textbf{Step \bm{$k~(2\le k \le d_1)$}}
    For $|\alpha_{y} \rangle \in \mathcal{B}_{12,13,14}$ with $i=k$, we obtain
    \begin{equation*}
        \begin{aligned}
            \langle000| E |000\rangle =\langle 00k| E |00k\rangle =\langle 0k0| E |0k0\rangle =\langle k00| E |k00\rangle,~2\le k \le d_1.
        \end{aligned}
    \end{equation*}
    Then $\langle \alpha_{1}| I\otimes E |\alpha_{1} \rangle = \langle \alpha_{y}| I\otimes E |\alpha_{y} \rangle$ where $|\alpha_{y} \rangle \in \mathcal{B}_{2,6,9}$ with $j=k$ implies that
    \begin{equation*}
        \begin{aligned}
            \langle000| E |000\rangle = \langle ki0| E |ki0\rangle =\langle i0k| E |i0k\rangle =\langle 0ik| E |0ik\rangle,~1 \le i \le d_1-1, ~2\le k \le d_1.
        \end{aligned}
    \end{equation*}
    \textbf{Step \bm{$d_1+1$}}
    For $|\alpha_{y} \rangle \in \mathcal{B}_{15,16,17}$, we have
    \begin{equation*}
        \begin{aligned}
            \langle000| E |000\rangle =\langle00i| E |00i\rangle, ~d_1+1\le i \le d_4-1;\\
            \langle000| E |000\rangle =\langle0i0| E |0i0\rangle, ~d_1+1\le i \le d_3-1;\\
            \langle000| E |000\rangle =\langle i00| E |i00\rangle, ~d_1+1\le i \le d_2-1.
        \end{aligned}
    \end{equation*}
    Then $\langle \alpha_{1}| I\otimes E |\alpha_{1} \rangle = \langle \alpha_{y}| I\otimes E |\alpha_{y} \rangle$ where $|\alpha_{y} \rangle \in \mathcal{B}_{2,6,9}$ with $j\ge d_1+1$ means that
    \begin{equation*}
        \begin{aligned}
            \langle000| E |000\rangle =\langle ji0| E |mnl\rangle, ~1\le i \le d_1-1, ~d_1+1\le j \le d_2-1;\\
            \langle000| E |000\rangle =\langle i0j| E |mnl\rangle, ~1\le i \le d_1-1, ~d_1+1\le j \le d_3-1;\\
            \langle000| E |000\rangle =\langle 0ij| E |mnl\rangle, ~1\le i \le d_1-1, ~d_1+1\le j \le d_2-1,
        \end{aligned}
    \end{equation*}
    and $\langle \alpha_{1}| I\otimes E |\alpha_{1} \rangle = \langle \alpha_{y}| I\otimes E |\alpha_{y} \rangle$ where $|\alpha_{y} \rangle \in \mathcal{B}_{5}$ means that
    \begin{equation*}
        \begin{aligned}
            \langle000| E |000\rangle =\langle ji0| E |mnl\rangle,~d_2\le i \le d_3-1, ~1\le j \le d_1-1.
        \end{aligned}
    \end{equation*}
    For $|\alpha_{y} \rangle \in \mathcal{B}_{3,4,7,8,10,11}$, we have
    \begin{equation*}
        \begin{aligned}
            \langle000| E |000\rangle =\langle ji0| E |ji0\rangle, ~~&d_1\le i \le d_2-1, ~1\le j \le d_2-1;\\
            \langle000| E |000\rangle =\langle ji0| E |ji0\rangle,~~&d_2\le i \le d_3-1, d_1\le j \le d_2-1;\\
            \langle000| E |000\rangle =\langle i0j| E |i0j\rangle, ~~&~1\le i \le d_1-1, ~d_3\le j \le d_4-1;\\
            \langle000| E |000\rangle =\langle i0j| E |i0j\rangle,~~&d_1\le i \le d_2-1, ~1\le j \le d_4-1;\\
            \langle000| E |000\rangle =\langle 0ij| E |0ij\rangle, ~~&~1\le i \le d_1-1, ~d_2\le j \le d_4-1;\\
            \langle000| E |000\rangle =\langle 0ij| E |0ij\rangle, ~~&d_1\le i \le d_3-1, ~1\le j \le d_4-1.
        \end{aligned}
    \end{equation*}
    \textbf{Step \bm{$d_1+2$}}
    Based on the above results, with respect to $|\alpha_{y} \rangle \in \mathcal{B}_{18-36}$, we have
    \begin{equation*}
        \begin{aligned}
            \langle000| E |000\rangle =\langle ijk| E |ijk\rangle,~1 \le i \le d_2-1, ~1 \le j \le d_3-1, ~1 \le k \le d_4-1.
        \end{aligned}
    \end{equation*}

    Thus $E\propto I$. This completes the proof.
\end{proof}

\end{document}